\PassOptionsToPackage{pdfpagelabels=false}{hyperref} 
\documentclass[a4paper,11pt]{scrartcl}
\usepackage[colorlinks=TRUE,linkcolor=blue,citecolor=green]{hyperref} 

\usepackage{a4wide}
\usepackage[german, english]{babel}
\usepackage[utf8]{inputenc}
\usepackage{graphicx} 

\usepackage{amsfonts,amsmath,amssymb,amsthm}
\usepackage{enumerate}
\usepackage{csquotes} 
\makeatletter
\g@addto@macro{\thm@space@setup}{\thm@headpunct{:}}
\makeatother

\newcommand{\e}{\mathrm{e}} 
\newcommand{\dd}[1]{\, \mathrm{d} #1}
\newcommand{\ex}[1]{\, \mathrm{exp} \Big(#1 \Big)}

\renewcommand{\i}{{\mathbf{\mathfrak{i}}}}
\newcommand{\EV}[2][\!]{\mathbb{E}^{#1} \left[ #2 \right]} 
\newcommand{\Var}[2][\!]{\; {\mathbb{V}\mathrm{ar}}^{#1} \big[ #2 \big]}
\newcommand{\logn}[1]{\, \ln \left( #1 \right)}
\newcommand{\sign}[1]{\, \mathrm{sgn} \left( #1 \right)}

\newcommand{\T}{^\top} 
\newcommand{\Transp}[1]{\left(#1\right)^{\! \top}} 

\newcommand{\ceil}[1]{\left\lceil#1\right\rceil}

\DeclareMathOperator{\E}{\mathbb{E}}
\DeclareMathOperator{\PM}{\mathbb{P}}
\DeclareMathOperator{\QM}{\mathbb{Q}}

\DeclareMathOperator{\R}{\mathbb{R}}
\DeclareMathOperator{\C}{\mathbb{C}}
\DeclareMathOperator{\F}{\mathcal{F}}
\DeclareMathOperator{\I}{\mathcal{I}}

\DeclareMathOperator{\Cor}{\mathbb{C}\mathrm{or}}

\renewcommand{\T}{\cdot} 
\renewcommand{\Transp}[1]{\big( #1 \big) \cdot} 

\newcommand{\defeq}{\mathrel{\mathop:}=}

\usepackage{tocloft}
\usepackage[numbers]{natbib} 
\usepackage{setspace}	
\usepackage{array} 
\usepackage[table]{xcolor} 
\usepackage{arydshln}	
\newcolumntype{L}[1]{>{\raggedright\let\newline\\\arraybackslash\hspace{0pt}}m{#1}}
\newcolumntype{C}[1]{>{\centering\let\newline\\\arraybackslash\hspace{0pt}}m{#1}}
\newcolumntype{R}[1]{>{\raggedleft\let\newline\\\arraybackslash\hspace{0pt}}m{#1}}
\newcolumntype{H}{>{\setbox0=\hbox\bgroup}c<{\egroup}@{}} 
\usepackage{framed} 

\usepackage[format=plain,hangindent=.3 cm,font=small]{caption}

\newtheorem{theorem}{Theorem}

\newtheorem{lemma}[theorem]{Lemma}

\theoremstyle{definition}
\newtheorem*{remark}{Remark}

\newcommand{\titleinfo}{The affine inflation market models}
\newcommand{\authorinfo}{Stefan Waldenberger} 
\newcommand{\keywords}{market models, inflation options, affine processes}

\begin{document}
\selectlanguage{english}

\AtEndDocument{\bigskip \bigskip{\footnotesize%
  \textsc{Graz University of Technology, Institute of Statistics, NAWI Graz} \par
  Kopernikusgasse 24/III, 8010 Graz, Austria \par
  \textit{E-mail address: }{stefan.waldenberger@tugraz.at} \par
}}

\begin{center}
\begin{minipage}{.8 \textwidth}
\begin{center}
{\Large \bf \titleinfo} \\ \bigskip
{\large \textsc{\authorinfo}} \\  \smallskip
\end{center}
  \bigskip \\
\textsc{Abstract:} Interest rate market models, like the LIBOR market model, have the advantage that the basic model quantities are directly observable in financial markets. Inflation market models extend this approach to inflation markets, where zero-coupon and year-on-year inflation-indexed swaps are the basic observable products. For inflation market models considered so far closed formulas exist for only one type of swap, but not for both. The model in this paper uses affine processes in such a way that prices for both types of swaps can be calculated explicitly. Furthermore call and put options on both types of swap rates can be calculated using one-dimensional Fourier inversion formulas. Using the derived formulas we present an example calibration to market data. 
\end{minipage}
\vspace{.5 cm}
\end{center}





Arbitrage-free inflation models have first been rigorously introduced in \citet{JY03}. Since then several inflation models have been proposed. Similar to interest rate models one can distinguish between short rate models and market models. While short rate models in the spirit of \citet{JY03} aim at modeling the unobservable continuous nominal and real short rate, market models use discrete observable rates as the basis for modeling (see \citet{BB04, ME05}). These observable rates are the basis of liquidly traded inflation swaps, zero coupon inflation-indexed swaps and year-on-year inflation-indexed swaps. While there have been several extensions of these models (e.g. \citet{MM06,MM09}) all of these models suffer the problem that there exist analytical formulas for only one type of swap, but not both. The model in this paper leads to closed formulas for both types.

Based on the ideas in \citet{KPT11} one can use affine processes to describe analytically highly tractable models. Affine processes are Markov process, where the characteristic function is of exponentially affine form, i.e.
$$ \EV{\e^{u X_t} \vert X_s} = \e^{\phi(t-s,u) + \psi(t-s,u) X_s}. $$
The class of affine processes contains a large number of processes, e.g. every Lèvy process is affine. Using the longest-dated nominal zero coupon bond as a numeraire we can model the normalized bond prices as \enquote{exponential martingales} with respect to an affine process $X$. In this type of models we are not only able to price both types of inflation swaps, but can also derive semianalytical formulas for calls and puts on the underlying inflation rates, another liquidly traded inflation derivative.

The structure of this paper is as follows. The first part describes inflation markets and typical traded derivatives. Afterwards we outline the setup of an inflation market model. The second part introduces the affine inflation market model and derives pricing formulas for the introduced inflation derivatives. The third part provides a concrete model specification including the calibration to actual market data. In the appendix we collect the properties of affine processes needed in this paper. Furthermore we specify the affine processes used in the numerical section.

\section{Inflation markets} \label{sec:inflprod}
Denote the time $t$ price of the nominal zero coupon bond with maturity $T$ by $P(t,T)$ and consider an inflation index with time $t$ value $I(t)$. Typically inflation indexes are so-called consumer price indexes (CPI). To shorten notation we will use the term CPI synonymous for inflation index, nevertheless the reader can think of an arbitrary inflation index. The basic mathematical instruments in inflation-linked markets are so called inflation-linked zero coupon bonds (corresponding to zero coupon bonds for nominal interest rate markets).
An inflation-linked zero coupon bond with maturity $T$ is a bond paying $I(T)$ at time $T$. Denote its price by $P_{ILB}(t,T)$.

In actual markets governments issue inflation-linked coupon bonds. Such a bond pays a fixed coupon on the variable basis $I(T_k)/I(T_0)$ at some fixed number of predetermined dates $T_k \leq T$ (typically annually), where $T_0$ is the time of issue. Additionally to the coupons such a bond redeems at maturity $T$ with value $\max\{I(T)/I(T_0),1\}$. Such a bond can therefore be described as a combination of inflation-linked zero coupon bonds plus an included option with payoff $(1-I(T)/I(T_0))_+$. In general these bonds are issued with maturities of several years and inflation is positive. In this case the included option has little influence on the total price, which is why it is market practice to mostly ignore it. In particular if one ignores these options it is possible to strip inflation-linked zero coupon bond prices out of actually traded inflation-linked coupon bonds by the same methods used for nominal quantities.

Consider the quantity 
\begin{equation} 
P_R(t,T) :=  \frac{P_{ILB}(t,T)}{I(t)},\label{eq:realbonddef}
\end{equation}
which is called the price of a real zero-coupon bond. Note that this is not the price\footnote{Here we mean price in terms of money which has to be paid at a transaction. In fact it is essentially this number that is quoted on trading screens. However, in case such a bond is traded, the cash-flow is then the quoted number multiplied by the according index ratio.} of a traded asset, but a theoretical quantity. The usage of the term price is motivated by the fact that this quantity can be viewed as the price of a zero coupon bond in a fictitious economy, where everything is measured in terms of the inflation index $I(t)$\footnote{One could interpret $I(t)$ as a numeraire, but one has to be careful not to use this as a mathematical numeraire, since $I(t)$ is not actually traded.}. Given real zero-coupon bond prices continuously compounded real interest rates are defined by
$R(t,T) := -{ \logn{P_R(t,T)} }/{(T-t)}.$ Accordingly one can define real counterparts to other nominal quantities such as forward interest rates or the short rate. This quantities are sometimes used as a starting point for inflation option pricing models (see e.g. \citet{JY03}, \citet{ME05}).

Next to inflation-linked bond markets there exist several liquidly traded inflation-linked derivatives. First consider the forward price of the inflation index (forward CPI), i.e. the at time $t$ fixed value $\I(t,T)$, which at time $T$ can be exchanged against $I(T)$ without additional costs.  Since $P_{ILB}(t,T)$ is the current price of $I(T)$, 
the time $t$ forward CPI for maturity $T$ is
\begin{equation}
\I(t,T) := \frac{P_{ILB}(t,T)}{P(t,T)} . \label{eq:fwdCPIdef}
\end{equation}
In a zero coupon inflation-indexed swap (ZCIIS) two parties exchange the realized inflation  $\frac{I(T)}{I(t)}$ against a fixed amount $(1+K)^{T-t}$. ZCIIS are mostly traded for full year maturities M. For $T = t + M$ the value of such a payer swap can be expressed as
\begin{equation} \label{eq:ZCIIS} P(t,T) \left( \frac{\I(t,T)}{I(t)} - (1+K)^M \right). \end{equation} The rate $K$, for which equation \eqref{eq:ZCIIS} is zero is then called the ZCIIS rate $ZCIIS(t; M)$. These ZCIIS rates are quoted in the market for several full-year maturities. 
\begin{remark} Note that ZCIIS rates and inflation-linked bonds are closely related via \eqref{eq:fwdCPIdef} and  \eqref{eq:ZCIIS}. In reality this relationship is not observed. This is partly due to different creditworthiness of counterparties in bond and swap markets. A more detailed analysis of this difference can be found in \citet{FL10}. For model calibration one has to choose one market, usually the swap market.
\end{remark}

Next to ZCIIS there is a second important type of swap in inflation markets, the year-on-year inflation-indexed swaps (YYIIS). These swaps exchange the annualized inflation against a fixed rate $K$, i.e. consider an annually spaced tenor structure $T_k=t+k, k=0,\dots,M$.
The netted payment of a payer YYIIS at time $T_k$ is $\left( {I(T_k)}/{I(T_{k-1})} -1 \right)  - K.$
Hence the inflation leg consists of payoffs of the form $$\frac{1}{T-S}  \left( \frac{I(T)}{I(S)} -1 \right).$$
Denote the forward value of such a payoff, the annualized forward inflation rate, by $F_I(t,S,T)$. Then the value of a payer YYIIS with maturity $M$ and strike $K$ can be expressed as
$$\sum_{k=1}^M P(t,T_k) (F_I(t,T_{k-1},T_k) -K).$$
The YYIIS rate $YYIIS(t; M)$ is the rate $K$ such that the corresponding YYIIS has zero value. Note that given YYIIS rates for all annual maturities one can calculate annual forward inflation rates $F_I(t,T_{k-1},T_k)$ and vice versa. 

Foward CPIs $\I(t,T)$, respectively forward inflation rates $F_I(t,S,T)$ are the mathematical quantities underlying the market-traded ZCIIS, respectively YYIIS. Inflation market models aim at modeling these quantities. In existing inflation market models either $\I(t,T)$ or $F_I(t,S,T)$ can be expressed by analytical formulas, but not both. Consider a market, where price processes are assumed to be semimartingales on a filtered probability space $(\Omega, \mathcal{A}, (\mathcal{F}_t), \PM)$. Fix a $T$-forward measure $\QM^T$, i.e. a probability measure equivalent to $\PM$ such that asset prices normalized with the numeraire price $P(t,T)$ are $\QM^T$-martingales. Then
\begin{align*}
\I(t,T) & = \EV[\QM^T]{I(T) \vert \F_t} \\
F_I(t,S,T) & =  \E^{\QM^T} \left[  \frac{1}{T-S}  \left( \frac{I(T)}{I(S)} -1 \right)  \bigg \vert \F_t \right]. \label{eq:fwdinfl} \end{align*}
Calculating the expectations of the inflation index, as well as the fraction of the inflation index at two different times proves difficult. In the model of this paper both are exponentially affine in the underlying driving stochastic process. For an affine process (see appendix) such expectations can be calculated and we are able to give semianalytical formulas for many standard options like caps and floors of forward inflation rates
.

\subsection{The inflation market model}
We now introduce the general setup of an (inflation) market model. 
Consider a tenor structure $0 < T_1 < \dots < T_N =: T$ and a market consisting of zero coupon bonds with maturities $T_k$ and prices $P(t,T_k)$. The price processes $(P(t,T_k))_{0 \leq t \leq T_k}$ are assumed to be positive semimartingales on a filtered probability space $(\Omega, \mathcal{A},  (\mathcal{F}_t)_{0 \leq t \leq T}, \PM)$, which satisfy $P(T_k,T_k) = 1$ almost surely.  If there exists an equivalent probability measure $\QM^{T}$ such that the normalized bond price processes $P(\cdot,T_k) / P(\cdot,T)$ are martingales\footnote{One can extend bond price processes to $[0,T]$ by setting $P(t,T_k) := \frac{P(t,T)}{P(T_k,T)}$ for $t > T_k$, so that $P(\cdot,T_k) / P(\cdot,T)$ is a martingale on $[0,T]$ if and only if it is a martingale on $[0,T_k]$. Economically this can be interpreted as immediately investing the payoff of a zero coupon bond into the longest-running zero coupon bond.
}, the market is arbitrage-free. This setup describes the class of interest rate market models like the classical LIBOR market model (\citet{BGM97}) and its extensions. 

To extend this setup to inflation markets consider an inflation index $I$, where we assume w.l.o.g. that $I(0)=1$. Assume there exist inflation-linked zero-coupon bonds with the same maturities\footnote{The assumption that for each zero coupon maturity there is a ILB with the same maturity is used only for notional convenience. 
} $T_1, \dots, T_N$ and price processes $P_{ILB}(t,T_k)_{0 \leq t \leq T_k}$, all of which are positive semimartingales\footnote{The inflation index is only described through the bond prices $P_{ILB}$. I.e. the distribution of $I(t)$ is only given at times $T_k$, where it coincides with the distribution of $P_{ILB}(T_k,T_k)$.}. If there exists an equivalent probability measure $\QM^T$ such that all normalized price processes
\begin{equation} \label{eq:martingaleassets}
\left( \frac{ P(t,T_k)}{P(t,T)} \right)_{0 \leq t \leq T_k}, \qquad \left(  \frac{ P_{ILB}(t,T_k)}{P(t,T)} \right)_{0 \leq t \leq T_k}
\end{equation}
are $\QM^T$-martingales, the extended market model is arbitrage-free. For given $\QM^T$ define the $T_k$-forward measures $\QM^{T_k}$ by
\begin{equation} \frac{\dd{\QM^{T_k}}}{\dd{\QM^{T}}} =  \frac{1}{P(T_k,T)} \frac{P(0,T)}{P(0,T_k)}.  \label{eq:measurechange}
\end{equation}
Under $\QM^{T_k}$ the forward interest rate
$$F^k(t) := \frac{1}{\Delta_k} \left(\frac{P(t,T_{k-1})}{P(t,T_k)} -1 \right), \qquad \Delta_k := T_k - T_{k-1}, $$
the earlier introduced forward CPI $$\I(t,T_k) = \frac{P_{ILB}(t,T_k)}{ P(t,T_k)},$$ 
and for $j < k$ the forward inflation rates $F_I(t,T_j,T_k)$ given by
$$1 + (T_k-T_j) F_I(t,T_j,T_k) = \EV[\QM^{T_k}]{\frac{I(T_k)}{I(T_j)} \vert \F_t}$$
are all martingales. Modeling (some of) these martingales is the starting point of inflation market models in the literature (see e.g. \citet{ME05}). In contrast we start by modeling the normalized bond prices in \eqref{eq:martingaleassets} and derive the above quantities thereof.

\section{The Affine inflation market model} \label{sec:affineCPI}

Let $(X_t)_{0 \leq t \leq T}$ with $X_0 = x$ be an analytic affine process with state space $\R^m_{\geq 0} \times \R^n$, $m > 0$, $n \geq 0$ on the probability space $(\Omega,\mathcal{A},(\F_t)_{0 \leq t \leq T},\QM^T)$ and define for $k=1,\dots,N$
\begin{equation} \label{eq:affineFwdCPImodel}
\begin{aligned}
 \frac{ P(t,T_k)}{P(t,T)} & := M_t^{u_k}, \qquad && u_k \in (\R_{\geq 0}^m \times \{0\}^n ) \cap \, \mathcal{V},  \\
 \frac{ P_{ILB}(t,T_k)}{P(t,T)} & := M_t^{v_k}, && v_k \in \R^{m+n} \cap \; \mathcal{V},
\end{aligned}
\end{equation}
where
\begin{equation} M_t^u := \EV[\QM^T]{\e^{u \cdot X_T} \vert \F_t} = \ex{\phi_{T-t}(u) + \psi_{T-t}(u) \cdot X_t},  \qquad u \in \mathcal{V}, \label{eq:affinemartingale} \end{equation}
with $\mathcal{V}$ defined in \eqref{eq:momset}. 
The processes $M_t^u$ are $\QM^T$-martingales by the definition of an affine process. Hence this model is arbitrage-free. 
Note that in \eqref{eq:affineFwdCPImodel} the parts of $u_k$ corresponding to real-valued components of $X$ are chosen to be zero. For a decreasing sequence $u_1 \geq \dots \geq u_N \geq 0$ one then has $M_t^{u_{k-1}} \geq M_t^{u_k}$, so that forward interest rates $F^k(t)$ are guaranteed to be nonnegative 
for all $k$. Contrary to interest rates\footnote{Although interest rates are currently negative in certain countries, interest rates are still bounded below by the costs of physically keeping money. We can incorporate bounds different from $0$ by setting $\frac{P(t,T_k)}{P(t,T)} := c_k M_t^{u_k}$. } inflation rates are not required to be nonnegative, which is why we do not restrict $v_k$ in \eqref{eq:affineFwdCPImodel}.
The values of $u_k$ and $v_k$ should be calibrated to fit the initial term structures, i.e. $M_0^{u_k} = P(0,T_k) / P(0,T)$ and $M_0^{v_k} = P_{ILB}(0,T_k) / P(0,T).$ By Lemma \ref{lem:ufitting} it follows that parameters $u_k$ fitting a current term structure with nonnegative forward interest rates can always be chosen to be decreasing. 
For multidimensional affine processes such sequences are far from unique. Concrete specifications how to choose $u_k$ and $v_k$ will be presented in section \ref{sec:numeric}.

The big advantage of this setup is that \eqref{eq:measurechange} in this case reads 
\begin{equation} \frac{\dd{\QM^{T_k}}}{\dd{\QM^{T}}} =  \frac{M_{T_k}^{u_k}}{M_{0}^{u_k}} = \frac{1}{M_{0}^{u_k}} \ex{\phi_{T-T_k}(u_k) + \psi_{T-T_k}(u_k)\T X_{T_k}}
\end{equation}
which is exponentially affine in $X$. In particular it is easy to check (see \citet{KPT11}) that for $0 \leq s \leq r$ and $\psi_{T-r}(u_k) + w \in \mathcal{V}$
\begin{equation}
\begin{aligned}
\EV[{\QM^{T_k}}]{\e^{w \cdot X_r} \vert \F_s}  =&  \ex{ \phi_{r-s}(\psi_{T-r}(u_k) + w) - \phi_{r-s}(\psi_{T-r}(u_k)) } \\
& \ex{\big( \psi_{r-s}(\psi_{T-r}(u_k) + w) - \psi_{r-s}(\psi_{T-r}(u_k)) \big) \cdot  X_s}. \label{eq:measurechar}
\end{aligned}
\end{equation}
Hence the moment generating function of $X$ is also known\footnote{This also shows that $X$ is a time-inhomogeneous affine process under $\QM^{T_k}$.} under different measures $\QM^{T_k}$.
Together with the exponential affine form of basic quantities  this is the reason why this model is analytically highly tractable. For example the forward rates $F^k$ satisfy 
$$(1+\Delta_k F^k(t)) = \frac{M_t^{u_{k-1}}}{M_t^{u_k}} = \e^{A(t,u_{k-1},u_k)+B(t,u_{k-1},u_k) \cdot X_t},$$ with
\begin{equation} \label{eq:AB}
\begin{aligned}
A(t,v,u) & := \phi_{T-t}(v) - \phi_{T-t}(u), \\
B(t,v,u) & := \psi_{T-t}(v) - \psi_{T-t}(u).
\end{aligned}
\end{equation}
Hence the $\QM^{T_k}$-extended moment generating function of $\logn{1+\Delta_k F^k(t)}$ can be calculated explicitly using \eqref{eq:measurechar}. 
The price of a caplet then follows using a Fourier-inversion formula (see \citet{KPT11}). Swaptions can also be dealt with (\citet{KPT11,GPSS14}) and so the most common interest rates derivatives can be calculated efficiently. We can use similar methods for inflation derivatives.

\subsection{Forward CPI and CPI options}
%
As mentioned before, the main advantage of this model is that for several important quantities the moment generating function is known under all forward measures $\QM^{T_k}$. Start by looking at the forward CPI
\begin{equation} \label{eq:fwdCPIaffine}
\I(t,T_k) = \frac{ P_{ILB}(t,T_k)}{P(t,T_k)} =  \frac{ P_{ILB}(t,T_k)}{P(t,T)}  \frac{ P(t,T)}{P(t,T_k)} = \frac{M_t^{v_k}}{M_t^{u_k}} =  \e^{A(t,v_k,u_k) + B(t,v_k,u_k) \cdot X_t},
\end{equation}
with $A$ and $B$ defined in \eqref{eq:AB}. Hence the forward CPI is of exponential affine form and therefore the $\QM^{T_k}$-moment generating function of its logartihm can be calculated using \eqref{eq:measurechar}. In particular, setting $A_I^k := A(T_k,v_k,u_k),  B_I^k := B(T_k,v_k,u_k)$ and using $I(T_k) = \I(T_k, T_k)$ one has
\begin{align*}
\mathcal{M}^{\QM^{T_k}}_{\logn{I(T_k)} \vert \F_s}(z) \defeq & \; \EV[\QM^{T_k}]{I(T_k)^z \vert \F_s} =  \EV[\QM^{T_k}]{\ex{z A_I^k + z B_I^k\T X_{T_k} }\vert \F_s} \\
 = & \ex{z  A_I^k + \phi_{T_k -s}(\psi_{T-T_k}(u_k) + z B_I^k) - \phi_{T_k -s}(\psi_{T-T_k}(u_k)) } \\
& \ex{\big(\psi_{T_k -s}(\psi_{T-T_k}(u_k) + z B_I^k)-\psi_{T_k -s}(\psi_{T-T_k}(u_k)) \big)\T X_s} \\
 = &\ex{z \phi_{T-T_k}(v_k) + (1-z)  \phi_{T-T_k}(u_k)} \\
&   \ex{ \phi_{T_k -s}\big(z \psi_{T-T_k}(v_k) + (1-z) \psi_{T-T_k}(u_k)\big)} \\
& \ex{ \psi_{T_k -s}\big( z \psi_{T-T_k}(v_k) + (1-z) \psi_{T-T_k}(u_k) \big)\T X_s} / M_s^{u_k}.
\end{align*}
Here the last equality follows using \eqref{eq:semiflow}. Note that this function is well defined and analytic in $z$ if $z \psi_{T-T_k}(v_k) + (1-z) \psi_{T-T_k}(u_k) \in  \mathrm{int}(\mathcal{V})$.

Given the moment generating function of $\logn{I(T_k)}$ CPI calls and puts can be calculated using the following well-known Fourier inversion formula (see e.g. \citet{EGP10}). 
If $R \in (1,\infty)$ such that $\mathcal{M}_{X \vert \F}(R) < \infty$, then
\begin{equation}  \label{eq:FourierCall}
\begin{aligned}
\EV{(e^X - K)_+ \vert \F} & = \frac{K}{\pi} \int_0^{\infty} \mathrm{Re} \left(  { \mathcal{M}_{X \vert \F}(\i u + R)} \frac{K^{-(\i u+R)}}{(\i u +R)(\i u +R -1)} 
\right) \dd{u}.
\end{aligned}
\end{equation}
Thus the price of a forward CPI call with maturity $T_k$ and payoff $(I(T_k) - K)_+$ is
\begin{equation*}
\mathrm{CPICall}(t,T_k,K) = \frac{{K} P(t,T_k)}{\pi} \int_0^{\infty} \mathrm{Re} \left(  { \mathcal{M}^{\QM^{T_k}}_{\logn{I(T_k)} \vert \F_s}(\i u + R)} \frac{{K}^{-(\i u+R)}}{(\i u +R)(\i u +R -1)} \right) \dd{u},
\end{equation*}
where $R>1$ is chosen to satisfy $R \psi_{T-T_k}(v_k) + (1-R) \psi_{T-T_k}(u_k) \in \mathrm{int}(\mathcal{V})$.
\begin{remark}
In section \ref{sec:inflprod} it is mentioned that ILBs usually come with an included option guaranteeing a redemption of at least the original nominal amount. For an ILB issued at $S$ and with maturity $T_k$ this translates into an option $(1 - I(T_k)/I(S))_+$ which corresponds to $1/I(S)$ CPI puts with strike $I(S)$. 
\end{remark}
 
\subsection{Forward Inflation and inflation caplets and floorlets}
Typically inflation market models are not able to handle both forward CPI and forward inflation products analytically. With the approach presented here forward inflation rates are of a similar form as forward CPIs. The annualized inflation $F_I(T_k,T_{k-j},T_k)$ satisfies
\begin{align}
1+ (T_k-T_{k-j}) F_I(T_k,T_{k-j},T_k) = \frac{I(T_k)}{I(T_{k-j})} = \e^{A_I^k + B_I^k \cdot X_{T_k} - A_I^{k-j} - B_I^{{k-j}} \cdot X_{T_{k-j}} } =: \e^{Y^k}. \label{eq:Yk}
\end{align}

The following Lemma gives the moment generating functions for random variables of this type. 
\begin{lemma} \label{lem:daffine}
Let $s \leq r \leq t \leq T$ and $$\psi_{T-t}(u_k)+w \in \mathcal{V} \text{ and } \psi_{t-r}(\psi_{T-t}(u_k)+ w) -  \psi_{T-r}(u_k) + u \in \mathcal{V}.$$ Then
\begin{equation}
\begin{aligned}
& \EV[\QM^{T_k}]{\e^{u\T X_r + w\T X_t} \vert \F_s }  =  \ex{\big(\psi_{r-s}(\psi_{t-r}(\psi_{T-t}(u_k)+w) + u) - \psi_{T-s}(u_k) \big)\T X_s} \\
& \qquad  \ex{\phi_{t-r}(\psi_{T-t}(u_k)+w) + \phi_{r-s}(\psi_{t-r}(\psi_{T-t}(u_k)+w)+ u) - \phi_{t-s}(\psi_{T-t}(u_k)) } .
\end{aligned}
\end{equation}
\end{lemma}
\begin{proof}
Using the tower property and applying \eqref{eq:measurechar} twice it follows
\begin{align*}
\E^{\QM^{T_k}}  \Big[ & \e^{u\T X_r + w\T X_t} \big \vert \F_s \Big] =  \EV[\QM^{T_k}]{\EV[\QM^{T_k}]{\e^{w\T X_t} \big \vert \F_r}  \e^{u\T X_r} \big \vert \F_s} \\
= & \E^{\QM^{T_k}} \Big[ \ex{ \phi_{t-r}(\psi_{T-t}(u_k)+w) -  \phi_{t-r}(\psi_{T-t}(u_k))} \\
&  \ex{\big( \psi_{t-r}(\psi_{T-t}(u_k)+w) -  \psi_{t-r}(\psi_{T-t}(u_k)) + u \big)\T X_r} \big \vert \F_s \Big] \\
=  & \ex{ \phi_{t-r}(\psi_{T-t}(u_k)+w) -  \phi_{t-r}(\psi_{T-t}(u_k)) } \\
& \ex{\phi_{r-s}(\psi_{t-r}(\psi_{T-t}(u_k)+w) + u) - \phi_{r-s}(\psi_{T-r}(u_k)) } \\
& \ex{\big(\psi_{r-s}(\psi_{t-r}(\psi_{T-t}(u_k)+w) + u) - \psi_{r-s}(\psi_{T-r}(u_k)) \big)\T X_s} \\
= & \ex{ \phi_{t-r}(\psi_{T-t}(u_k)+w) + \phi_{r-s}(\psi_{t-r}(\psi_{T-t}(u_k)+w)+ u) - \phi_{t-s}(\psi_{T-t}(u_k)) } \\
& \ex{\big(\psi_{r-s}(\psi_{t-r}(\psi_{T-t}(u_k)+w) + u) - \psi_{T-s}(u_k) \big)\T X_s}.
\end{align*}
Here we used the semiflow property \eqref{eq:semiflow} to simplify the expression.
\end{proof} 

By Lemma \ref{lem:daffine} the $\QM^{T_k}$-moment generating function of $Y^k$ defined in \eqref{eq:Yk} is
\begin{align*}
\mathcal{M}^{\QM^{T_k}}_{Y^k \vert \F_s}(z) 
 = & \EV[\QM^{T_k}]{ \ex{z A_I^k + z  {B_I^k} \cdot X_{T_k} - z A_I^{k-j} - z  {B_I^{k-j}} \cdot X_{T_{k-j}} } \vert \F_s } \\
= &  \ex{z A_I^k - z A_I^{k-j} +  \phi_{T_k-T_{k-j}}(\psi_{T-T_k}(u_k)+z B_I^k) } \\
& \ex{\phi_{T_{k-j}-s}(\psi_{T_k-T_{k-j}}(\psi_{T-T_k}(u_k)+z B_I^k) - z B_I^{k-j}) - \phi_{T_k-s}(\psi_{T-T_k}(u_k)) } \\
& \ex{\Transp{\psi_{T_{k-j}-s}(\psi_{T_k-T_{k-j}}(\psi_{T-T_k}(u_k)+ z B_I^k) -  z B_I^{k-j}) - \psi_{T-s}(u_k)} X_s},
\end{align*}
which is well-defined if \begin{equation} \left \{ \psi_{T-T_k}(u_k)+z B_I^k, \psi_{T_k-T_{k-j}}(\psi_{T-T_k}(u_k)+z B_I^k) -  \psi_{T-T_{k-j}}(u_k) - z B_I^{k-j} \right \} \subset \mathcal{V}. \end{equation}
The forward inflation rate is then given by $1+(T_k-T_{k-j}) F_I(t,T_{k-j},T_k) = \mathcal{M}^{\QM^{T_k}}_{Y^k \vert \F_t}(1)$.
So for $\psi_{T-T_{k-j}}(v_{k})-B_I^{k-j} \in \mathcal{V}$ it is 
\begin{align*}
1+(T_k-T_{k-j}) & F_I(t,T_{k-j},T_k) =
 \ex{\Transp{\psi_{T_{k-j}-t}(\psi_{T-T_{k-j}}(v_{k})-B_I^{k-j})-\psi_{T-t}(u_k)}X_t} \\
& \ex{\phi_{T-T_{k-j}}(u_{k-j}) + \phi_{T_{k-j}-t}(\psi_{T-T_{k-j}}(v_{k})-B_I^{k-j} )+ \phi_{T-t}(u_k)}.
\end{align*}
Furthermore the payoff of an inflation caplet with strike $K$ is 
\begin{align*}
(T_k - T_{k-j}) (F_I(T_k,T_{k-j},T_k)-K)_+ =  \left( \frac{I(T_k)}{I(T_{k-j})} - \tilde{K} \right)_+
\end{align*}
where $\tilde{K} = 1 + (T_k - T_{k-j}) K$. With the Fourier inversion formula \eqref{eq:FourierCall} one can calculate the price of an inflation caplet. In particular, we have for $R>1$ 
\begin{align*}
\mathrm{InflCpl}(t,T_{k-j},T_k,K) = \frac{\tilde{K} P(t,T_k)}{\pi} \int_0^{\infty} \mathrm{Re} \left(  { \mathcal{M}^{\QM^{T_k}}_{Y^k \vert \F_t}(\i u + R)} \frac{\tilde{K}^{-(\i u+R)}}{(\i u +R)(\i u +R -1)} \right) \dd{u},
\end{align*}
provided that $\psi_{T-T_k}(u_k)+R B_I^k \in \mathrm{int}(\mathcal{V})$ and $\psi_{T_k-T_{k-j}}(\psi_{T-T_k}(u_k)+R B_I^k) -  \psi_{T-T_{k-j}}(u_k) - R B_I^{k-j} \in \mathrm{int}(\mathcal{V})$.

\subsection{Correlation} \label{sec:correlation}
So far we considered the pricing of typical market traded options. Another important aspect is the correlation structure. The relevant quantities
\begin{equation} \label{eq:logquants}
\begin{aligned}
\logn{1 + \Delta_k F_n^k(t)}  & =  \logn{\frac{M_t^{u_{k-1}}}{M_t^{u_k}}} ={A(t,u_{k-1},u_k) + B(t,u_{k-1},u_k) X_t}, \\
\logn{\I^j(t)} & =  \logn{ \frac{M_t^{v_j}}{M_t^{u_j}} } ={A(t,v_j,u_j) + B(t,v_j,u_j) X_t}. 
\\ \logn{1 + \Delta_k F_I(t,T_{k-j},T_k)} & = \text{const} + 
\big(\psi_{T_{k-j}-t}(\psi_{T-T_k}(v_k)  - B_I^{k-j}) - \psi_{T-t}(u_k) \big)\T X_t. 
\end{aligned}
\end{equation}
are all affine transformation of $X_t$ and the correlation for two such terms is 
\begin{align*}
\Cor [A_t + B_t \cdot X_t,\tilde{A}_t + \tilde{B}_t \cdot X_t] & = \frac{\Var{B_t \cdot X_t,\tilde{B}_t \cdot X_t }}{\sqrt{\Var{B_t \cdot X_t}}\sqrt{\Var{\tilde{B}_t \cdot X_t }}}.
\end{align*}
For independent components of $X_t$ this simplifies to\footnote{Up to some technical conditions the variance of a one-dimensional affine process is
$$\Var{X_t^i} =  \left. \frac{\partial^2}{\partial^2 u} \right \vert_{u=0}  (\phi_t^i(u) + \psi_t^i(u) X_0^i) .$$}
\begin{align} \label{eq:corrstructure}
 \frac{\sum_{i=1}^d B_t^i \tilde{B}_t^i \Var{X_t^i}}{\sqrt{\sum_{i=1}^d (B_t^i)^2 \Var{X_t^i}}\sqrt{\sum_{i=1}^d (\tilde{B}_t^i)^2 \Var{X_t^i} }}.
\end{align}
Hence correlations strongly depend on $B(t,u_{k-1},u_k) = \psi_{T-t}(u_{k-1}) - \psi_{T-t}(u_k)$ and $B(t,v_j,u_j) = \psi_{T-t}(v_j) - \psi_{T-t}(u_j)$, respectively the structure of the $v_k$ and $u_k$. The exact correlation depends on the used measure (e.g. $\QM^{T_k}, \PM$), but choosing $v_k$ and $u_k$ cleverly, one can guarantee that the correlation structure, i.e. the correlation signs, stay the same. Concrete specifications for meaningful correlation structures will be given in the next section. Similar observations can also be made for instantaneous correlations of the corresponding quantities in the case of continuous affine processes (see \citet{GPSS14} for the general idea). 

\section{Implementation example} \label{sec:numeric}  \label{sec:calibrationexample}
We design the structure of the affine inflation market model in such a way that the calibration can be separated into the calibration to nominal market data and the calibration to inflation market data afterwards. The method used to calibrate to nominal market data is based on the ideas in \citet{GPSS14}. There they fit a multiple curve affine LIBOR market model by using a common driving process $X^0$ plus additional driving processes $X^1, \dots, X^M$, all of which are independent, affine and nonnegative. 
For calibration they use caplets with full-year maturities and underlying forwards of tenors less than a year. Using one individual driving process for each year one can then use an iterative procedure to calibrate to market data. Their approach can also be used in this setup. 

In particular consider a semiannual tenor structure $T_k = k/2, k=1,\dots,N$, $N$ even, and a driving affine process consisting of $M+1 = N/2 +1$ components $X^0, X^1, \dots ,X^M$, all of which are independent analytic affine processes with functions $\phi^i$ and $\psi^i$, $i=0,\dots,M$. Then by \eqref{eq:affinecombphipsi}
\begin{align*}
\phi_t(u) & = \sum_{i=0}^M \phi_t^i(u^i) \\
\psi_t(u) & = (\psi_t^0(u^0),\psi_t^1(u^1), \dots, \psi_t^M(u^M)),
\end{align*}
where $u^i, i=0, \dots, M$ denotes the corresponding component of $u \in \R^{M+1}$. To describe an affine inflation market model we also have to specify the vectors $u_k$. The structure of the vectors should be so that the following points are satisfied.
\begin{itemize}
\item Forward interest rates are nonnegative. This is the case if $0 \leq u_k \leq u_{k-1}$. 
\item The model matches the initial interest rate term structure. This basically fixes one component of each vector $u_k$.
\item Calibration to market data is possible by using an iterative procedure. 
\item The model has a meaningful correlation structure.
\end{itemize}
\begin{table}
\begin{center}
\begin{tabular}{l | C{2.5em}:C{2.5em}C{2.5em}C{2.5em}C{2.5em}:C{3.4 em}C{3.4 em}C{2.5em}C{2.5em}  }
& $X^0$ & $X^1$ & $X^2$ & \dots & $X^{M}$ & $X^{M+1}$ & $X^{M+2}$ & \dots & $X^{2M}$ \\ \hdashline
$u_1$ & $\tilde{u}_1$ & $\overline{u}_1$ & $\overline{u}_3$ & $\dots$ & $\overline{u}_{N-1}$ & 0 & $0$ & $\cdots$ & $0$ \\
$u_2$ & $\tilde{u}_2$ & $\overline{u}_2$ & $\overline{u}_3$ & $\dots$ & $\overline{u}_{N-1}$ & 0 & $0$ & $\cdots$ & $0$ \\
$u_3$ & $\tilde{u}_3$ & 0 & $\overline{u}_3$ & $\dots$ & $\overline{u}_{N-1}$ & $0$ & 0  & $\cdots$ & $0$ \\
$u_4$ & $\tilde{u}_4$ & 0 & $\overline{u}_4$ & $\dots$ & $\overline{u}_{N-1}$ & $0$ & 0  & $\cdots$ & $0$ \\
\vdots & \vdots & \vdots & \vdots & $\ddots$ & \vdots & \vdots & \vdots & $\ddots$ & \vdots \\
$u_{N-2}$ & $\tilde{u}_{N-2}$ & 0 & 0 & $\dots$ & $\overline{u}_{N-1}$ & $0$ &0 & $\cdots$ & 0 \\
$u_{N-1}$ & $\tilde{u}_{N-1}$ & 0 & 0 & $\dots$ & $\overline{u}_{N-1}$ & $0$ &0 & $\cdots$ & 0 \\
$u_{N}$& $\tilde{u}_{N}$ & 0 & 0 & $\dots$ & $\overline{u}_{N}$ & $0$ &0 & $\cdots$ & 0 \\
\end{tabular}
\end{center}
\captionsetup{singlelinecheck=false, margin = .5 cm}
\caption{Description of the parameter structure $u_k$. Each row corresponds to one vector with the column names denoting the process the position in the vector corresponds to.}
\label{tab:upar}
\end{table}
This can be achieved by choosing the vectors $u_k$ in the following way. They depend on $2N$ real parameters 
$$\tilde{u}_1 \geq \dots \geq \tilde{u}_N \geq 0, \qquad \overline{u}_1 \geq \dots \geq \overline{u}_N \geq 0.$$
For $1 \leq j \leq M$ set (compare table \ref{tab:upar}, ignoring the zero columns in the rightmost side for now)
$$ u_k = \tilde{u}_k \e^0 + \overline{u}_k e^{\ceil{\frac{k}{2}}} + \sum_{l= \ceil{\frac{k}{2}}+1 }^M \overline{u}_{2l-1} e^l$$
where $e^0,e^1,\dots,e^M$ denote the base vectors $(1,0, \dots, 0), \dots, (0,\dots,0,1)$ of $\R^{M+1}$. Note that with this choice the vectors $(u_k)$ are decreasing. Given $\tilde{u}_1,\dots,\tilde{u}_n$ and the processes $X^0,X^1,\dots,X^M$ by Lemma \ref{lem:ufitting} the parameters $\overline{u}_1,\dots,\overline{u}_N$ are determined by fitting the current term structure. I.e. we require that $$\frac{P(0,T_k)}{P(0,T)} =  \EV[\QM^T]{\e^{u_k\T X_T}} = 
\EV[\QM^T]{\e^{\tilde{u}_k X^0_T}} \EV[\QM^T]{\e^{\overline{u}_k X^{\ceil{\frac{k}{2}}}_T}}  \prod_{i= \ceil{\frac{k}{2}}+1}^M \EV[\QM^T]{\e^{\overline{u}_{2l-1} X^l_T}},$$
so that the parameters $\overline{u}_1,\dots,\overline{u}_N$ can be calculated using backwards iteration.
Furthermore the semiannual forward interest rates with full-year maturities $F^{2k}$ only depend on $u_{2k-1},u_{2k}$ and the processes $X^0, X^k, \dots, X^{M}$. Hence if $X^0$ and $\tilde{u}_1, \dots, \tilde{u}_N$ are already specified, one can fit $X^{M}$ to caplets on the forward interest rate $F^{N}$ and then go backwards to iteratively fit the processes $X^k$ to caplets on forward interest rates $F^{2k}$. Hence if $X^0$ and  $\tilde{u}_1, \dots, \tilde{u}_N$ are fixed, all the remaining parameters can be calibrated to the yield curve and caplet prices. As stated in \citet{GPSS14} and confirmed by our numerical tests the concrete choice of $X^0$ and $\tilde{u}_k$ (within a meaningful range) has no qualitative impact on the resulting calibration quality. Henceforth $X^0$ is fixed as a CIR process (specifications of this process can be found in the appendix, see equation \eqref{eq:CIR}). So far we have said nothing about resulting correlations. This is where the choice of the parameters $\tilde{u}_1, \dots, \tilde{u}_n$ comes in.
The relevant functions for correlations of forward interest rates $F^{2k}$ are
\begin{equation*}
B(t,u_{2k-1},u_{2k}) = ( \psi_{T-t}^0(\tilde{u}_{2k-1}) - \psi_{T-t}^0(\tilde{u}_{2k}), 0, \dots, 0, \psi_{T-t}^{k}(\overline{u}_{2k-1})- \psi_{T-t}^{k}(\overline{u}_{2k}), 0, \dots, 0).
\end{equation*}
These functions are \enquote{orthogonal} except for the first component\footnote{Although only stated for even forward rates for notional simplicity this is true for all forward rates.}. Hence by equation \eqref{eq:corrstructure} the correlation structure mainly depends on the sequence $(\tilde{u}_k)$.
Since $X^0$ is nonnegative the function $\psi^0_t(u)$ is increasing in $u$ (see \citet{KPT11}). With $(\tilde{u}_k)$ being a decreasing sequence this results in nonnegative correlations. Setting $\tilde{u}_k = \tilde{u}$ for all $k$ results in zero correlation\footnote{Negative correlations in this setup are only possible if the sequence $(\tilde{u}_k)$ is not decreasing which means that forward interest rates can become negative.}. The magnitude of correlations depends on how much of the variance of forward bond prices is explained by $\tilde{u}_k X^0$. Here we choose the factors $\tilde{u}_k$ so that $\EV[\QM^T]{\e^{2 \tilde{u}_k X_T^c}} = P(0,T_k)/P(0,T)$. The idea behind this choice is that approximately half of the variance should be explained by the common factor. Alternative if one has additional information on correlations (e.g. through market data), this could be incorporated in $\tilde{u}_k$. 

For calibration we used market data from September, 29th 2011. The yield curve is bootstrapped from LIBOR and swap rates and is displayed in figure \ref{fig:yieldcurve}. The used caplet implied volatilities for 6-month forward interest rates with strikes $1\%$ to $6\%$ and maturities from 1 to 10 years are bootstrapped from cap data. The resulting implied volatilities can be found in figure \ref{fig:capletfit}. 
\begin{figure}
\centering
\includegraphics[width=0.8 \textwidth]{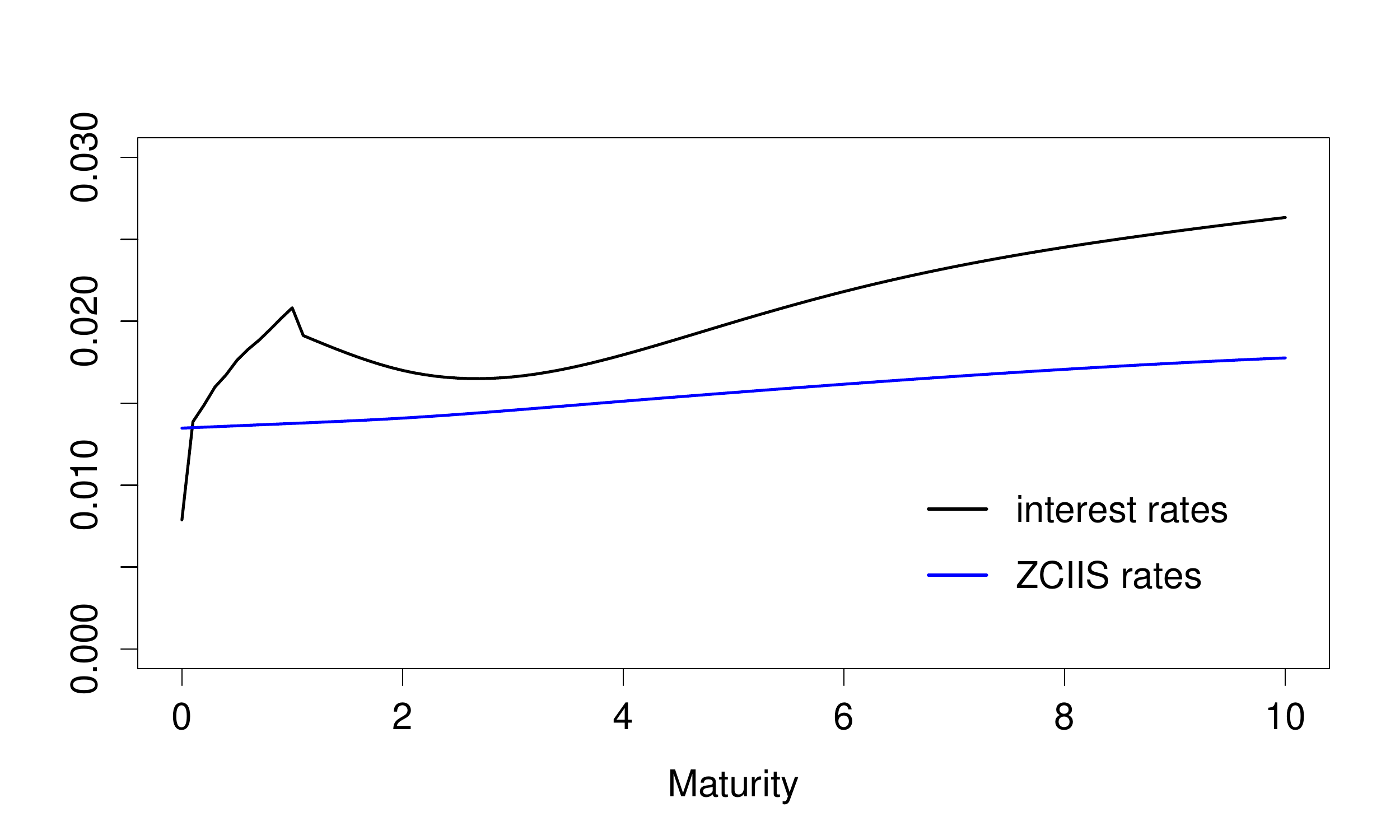}
\caption{Yield curve and ZCIIS curve from September, 29th 2011}
\label{fig:yieldcurve} \label{fig:ZCIIScurve}
\end{figure}
We fixed the time horizon $T=10$ and chose the parameters of the common CIR process as $\lambda=0.026,\theta=0.65,\eta=0.5,x=3.45$. The $M$ individual driving processes are chosen to be CIR processes with added jumps (see appendix, equation \eqref{eq:CIRGamOU}). Their parameters and the parameters $\overline{u}_k$ are calibrated with the mentioned recursive method. In each step the parameters are chosen so that the mean squared errors of implied volatilities are minimized. In contradiction with \citet{GPSS14} we were not able to produce a similar calibration quality as described in their paper\footnote{Several requests for clarification with the authors have resulted in the answer that their results are currently not in a state to be shared.}. The resulting calibration can be found in figure \ref{fig:capletfit}. Especially for long dated caplet volatilities the pronounced skew could not be reproduced. Nevertheless the model provides a reasonable fit of the caplet volatility surface, especially since the focus is on inflation derivatives. 
\begin{figure}
\centering
\includegraphics[width=0.55 \textwidth]{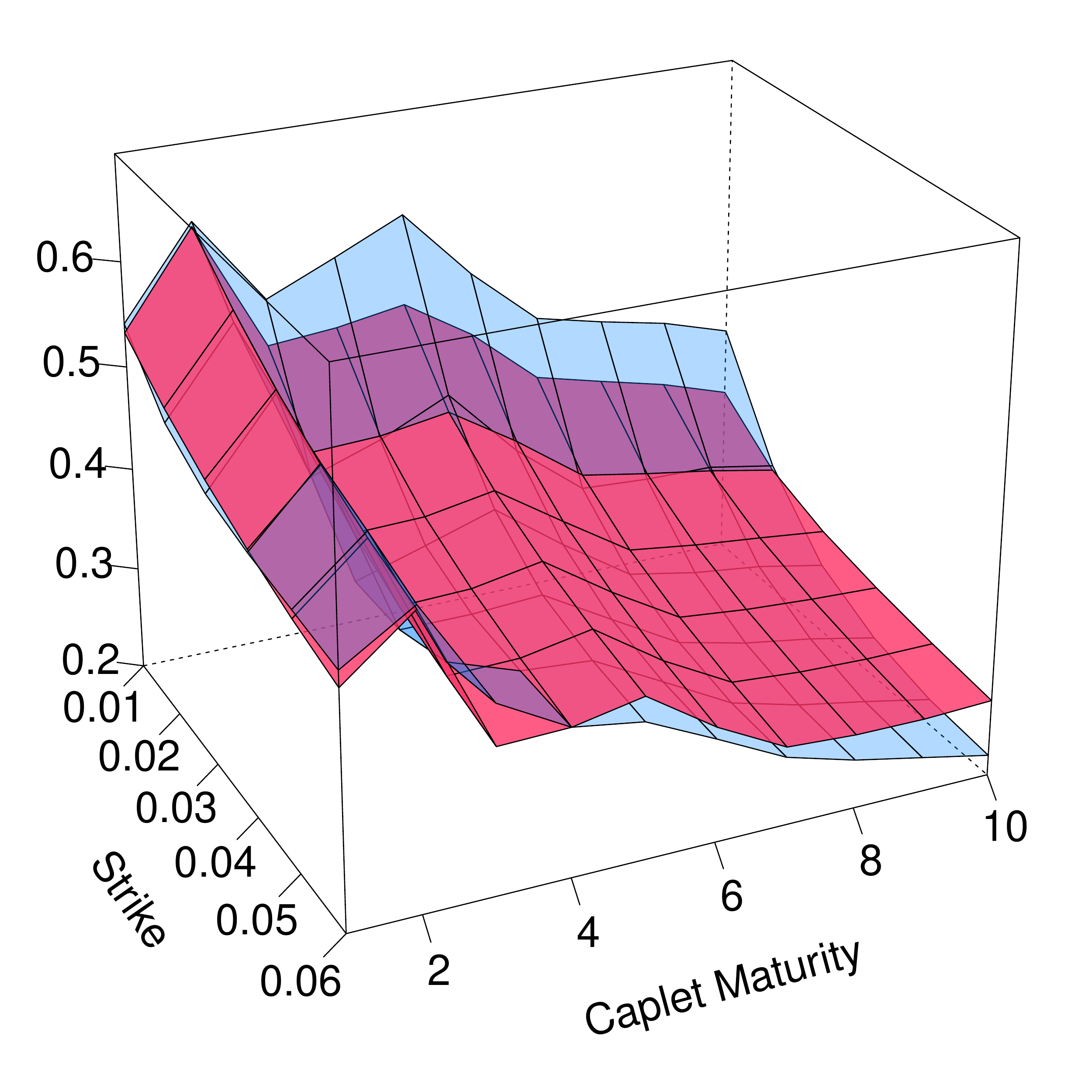}
\captionsetup{singlelinecheck=false, margin = 1.5 cm}
\caption{Market and model caplet implied volatilities of 6-month forward interest rates with strikes $1\%$ to $6\%$ and maturities from 1 to 10 years. Market volatilities are in a transparent blue, while model volatilities are displayed in red.}
\label{fig:capletfit}
\end{figure}

The next step is to extend the calibration to inflation markets. Additionally to the $M + 1$ driving processes used for modeling interest rates we use another $M$ independent analytic affine processes (one for each year) driving inflation-related quantities. Since the individual processes are independent by \eqref{eq:affinecombphipsi} the setup without inflation can be embedded by setting the additional components of $u_k$ equal to zero (see table \ref{tab:upar}). So from now on we assume that the processes $X^0,X^1,\dots, X^M$ and the vectors $u_1, \dots, u_N$ are already fixed.

The choice of the inflation parameters $v_k$ focuses on the same aspects aspects as the choice of $u_k$ without the restriction that inflation rates have to be nonnegative. We again assume that the vectors $v_k$ are determined by $2N$ parameters\footnote{We actually only require $N$ parameters, since the odd rows in table \ref{tab:uinfpar} do not contribute for the considered annual inflation rates. For notional simplicity we nevertheless consider $2N$ parameters.}
$\tilde{v}_1, \dots, \tilde{v}_N$, $\overline{v}_1, \dots, \overline{v}_N.$
In particular, we choose (see also table \ref{tab:uinfpar})
\begin{table}
\begin{center}
\begin{tabular}{l | C{2.5em}:C{2.5em}C{2.5em}C{2.5em}C{2.5em}:C{3.4 em}C{3.4 em}C{2.5em}C{2.5em} }
& $X^0$ & $X^1$ & $X^2$ & \dots & $X^{M}$ & $X^{M+1}$ & $X^{M+2}$ & \dots & $X^{2M}$ \\ \hdashline
\rowcolor{gray}
$v_1$ & $\tilde{v}_1$ & $\overline{u}_1$ & $\overline{u}_3$ & $\dots$ & $\overline{u}_{N-1}$ & $\overline{v}_1$ & $0$ & $\cdots$ & $0$ \\
$v_2$ & $\tilde{v}_2$ & $\overline{u}_2$ & $\overline{u}_3$ & $\dots$ & $\overline{u}_{N-1}$ & $\overline{v}_2$ & $0$ & $\cdots$ & $0$ \\
\rowcolor{gray} $v_3$ & $\tilde{v}_3$ & 0 & $\overline{u}_3$ & $\dots$ & $\overline{u}_{N-1}$ & $0$ & $\overline{v}_3$  & $\cdots$ & $0$ \\
$v_4$ & $\tilde{v}_4$ & 0 & $\overline{u}_4$ & $\dots$ & $\overline{u}_{N-1}$ & $0$ & $\overline{v}_4$  & $\cdots$ & $0$ \\
\vdots & \vdots & \vdots & \vdots & $\ddots$ & \vdots & \vdots & \vdots & $\ddots$ & \vdots \\
$v_{N-2}$ & $\tilde{v}_{N-2}$ & 0 & 0 & $\dots$ & $\overline{u}_{N-1}$ & $0$ &0 & $\cdots$ & 0 \\\rowcolor{gray} $v_{N-1}$ & $\tilde{v}_{N-1}$ & 0 & 0 & $\dots$ & $\overline{u}_{N-1}$ & $0$ &0 & $\cdots$ & $\overline{v}_{N-1}$ \\
$v_{N}$& $\tilde{v}_{N}$ & 0 & 0 & $\dots$ & $\overline{u}_{N}$ & $0$ &0 & $\cdots$ & $\overline{v}_{N}$ \\
\end{tabular}
\end{center}
\captionsetup{singlelinecheck=false, margin = .5 cm}
\caption{Description of the inflation parameter structure $v_k$. Each row corresponds to one vector with the column names denoting the process the position in the vector corresponds to. Note that for annual inflation rate option pricing the vectors $v_j$ with $j$ odd do not matter.}
\label{tab:uinfpar}
\end{table}
$$ v_k = \tilde{v}_k \e^0 + \overline{u}_k e^{\ceil{\frac{k}{2}}} + \sum_{l= \ceil{\frac{k}{2}}+1 }^M \overline{u}_{2l-1} e^l + \overline{v}_k e^{M+\ceil{\frac{k}{2}}}.$$
Choosing the nominal components $v_k^i = u_k^i$ for $i = 1, \dots, M$ has the advantage that forward CPIs and therefore also forward inflation rates do not depend on the nominal processes $X^1, \dots, X^M$. Together with the choice of $v_k$ with respect to the inflation processes $X^{M+1}, \dots X^{2M}$ this implies that the forward CPI $\I(t,T_{2k})$ depends only on $X^0$ and $X^{M+k}$. In particular the function $B(t,v_{2k},u_{2k})$ corresponding to the forward CPI $\I(t,T_{2k})$ defined in \eqref{eq:fwdCPIaffine} is
\begin{align*}
B(t,v_{2k},u_{2k}) & = ( \psi_{T-t}^0(\tilde{v}_{2k}) - \psi_{T-t}^0(\tilde{u}_{2k}), 0, \dots, 0, \psi_{T-t}^{M+k}(\overline{v}_{2k}), 0, \dots, 0).
\end{align*}
From this it follows that for different forward CPIs these functions are \enquote{orthogonal} except for the first component. Furthermore except for the first component they are also \enquote{orthogonal} to the functions $B(t,u_{j-1},u_j)$ relevant for forward interest rates. Hence the correlation structure mainly depends on $\tilde{u}_k$ and $\tilde{v}_k$. Since $\psi_t^0$ is monotonically increasing, one should choose $\tilde{v}_{k} >  \tilde{u}_{k}$ if the corresponding forward CPI should be positively correlated with nominal interest rates or $\tilde{v}_{k} < \tilde{u}_{k}$ if it should be negatively correlated\footnote{We decided to introduce correlation for the inflation part only via the common factor process. We could also have changed the parameters $\overline{u}_k$ in order to introduce correlation. In this case even more complicated correlation patterns could be created.}. Also note that two annual forward CPIs are positively correlated if $\sign{\tilde{v}_k - \tilde{u}_k} = \sign{\tilde{v}_j - \tilde{u}_j}$. This therefore gives us some criteria how to determine the parameters $\tilde{u}_k$ from correlation assumptions and from now on we assume that the parameters $\tilde{v}_1, \dots, \tilde{v}_N$ are given. 
Assuming also a fixed process $X^{M+k}$ we would like to determine $\overline{v}_k$ from the current term structure $ P_{ILB}(0,T_k) / P(0,T)$. We differentiate between two cases. First consider a nonnegative affine process $X^{M+k}$. 
In this case by Lemma \ref{lem:ufitting} there exists a unique $\overline{v}_k$, so that $M_0^{v_k} = P_{ILB}(0,T_k) / P(0,T)$. For $\tilde{v}_k \approx \tilde{u}_k$ this typically results in $\overline{v}_k > 0$. In this case zero coupon inflation for $[t,T_k]$ is always positive. To avoid this alternatively consider an affine process, which takes negative and positive values. In this case $M_0^v$ is not necessarily increasing in  $\overline{v}_k$. However, by Lemma \ref{lem:ufitting} it is still a convex function in $\overline{v}_k$. This means that there are at most two possible choices for $\overline{v}_k$, in which case we need to pick one. For the results in this paper we only used results where one choice was smaller than zero and the other one larger than zero, which we then picked.
To determine the processes $X^{M+1},\dots,X^{2M}$ notice that the annual forward inflation rate $F_I(t,T_{2k-2},T_{2k})$ depends only the processes $X^0, X^{M+k-1},X^k$ and $F_I(t,T_0,T_2)$ depends only on $X^0, X^{M+1}$. By starting with inflation options on $F_I(T_2,T_0,T_2)$ one can calibrate the parameters of $X^{M+1}$. Then one can iteratively calibrate the parameters of $X^{M+k}$ using inflation options on $F_I(T_{2k},T_{2(k-1)},T_{2k})$. 

For the calibration example we used ZCIIS rates from $1$ to $10$ years (see figure \ref{fig:ZCIIScurve}) and inflation options for strikes ranging from $-2\%$ to $6\%$, the prices of which are displayed in figure \ref{fig:inflfit}.
We chose $\tilde{v}_k = \tilde{u}_k (1+c k)$ with $c \approx 0.08$, so that $c k$  is between $1$ and $1.15$. This means that correlation are positive as usually observed. For the processes $X^{M+1}, \dots, X^{2M}$ we used Ornstein-Uhlenbeck processes with added jumps (see appendix, equation \eqref{eq:DGamOUBM}). 
The parameters of these processes and the sequence $(\overline{v}_k)$ are then calibrated as described. In each step the mean squared errors of option prices are minimized. The resulting fit is displayed in figure \ref{fig:inflfit} and this shows that the calibration is very accurate.
\begin{figure}
\centering
\includegraphics[width=0.55 \textwidth]{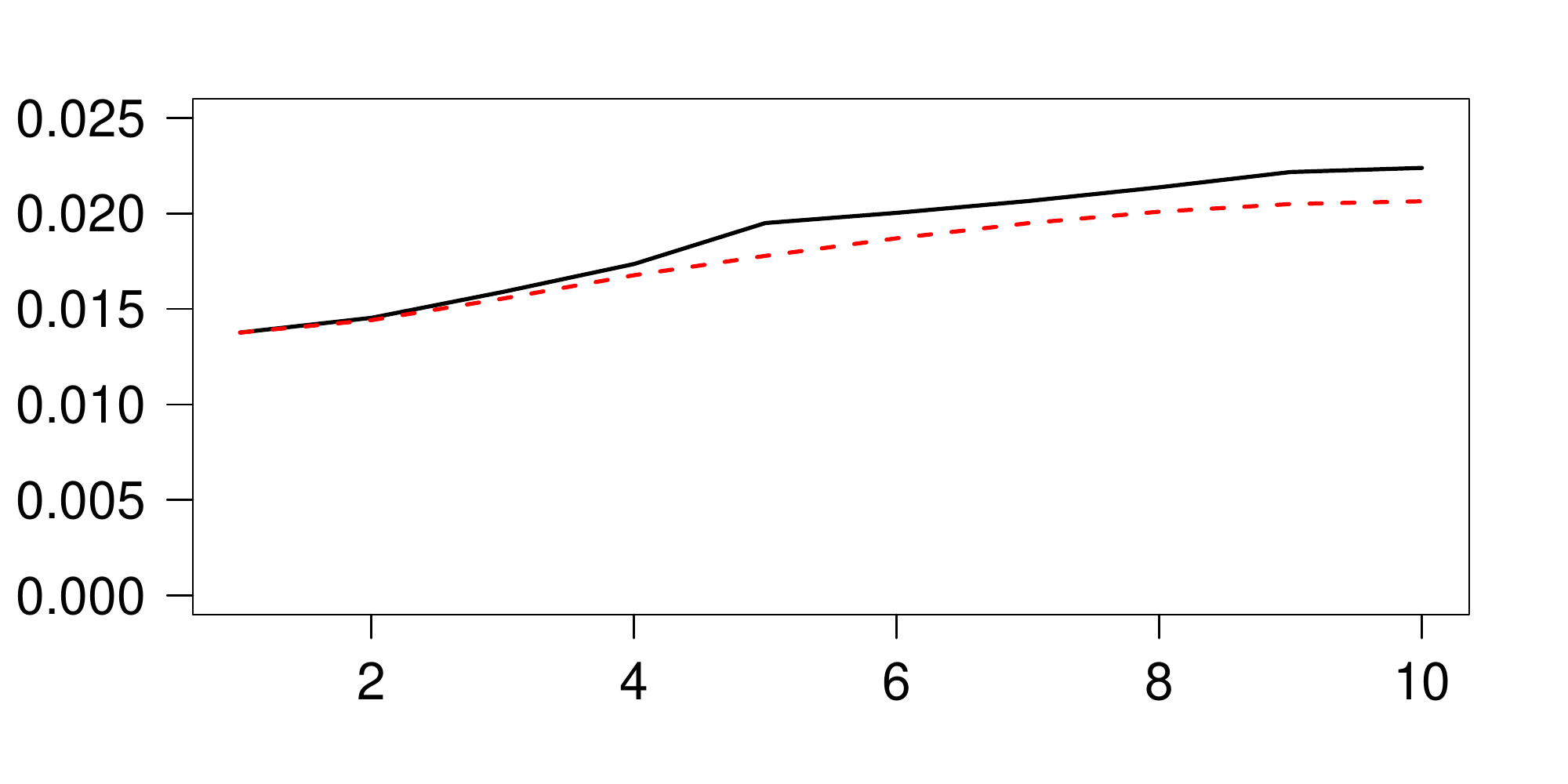}
\captionsetup{singlelinecheck=false, margin = 1.5 cm}
\caption{\small Linear interpolated annual forward inflation rates $F_I(0,T_{2(k-1)},T_{2k})$ (black) and its approximation $\I(t,T_{2k})/\I(t,T_{2(k-1)})-1$ (dashed red) for maturities from 1 to 10 years.} 
\label{fig:forwardinflation}
\end{figure}
\begin{figure}
\centering
\includegraphics[width=0.55 \textwidth]{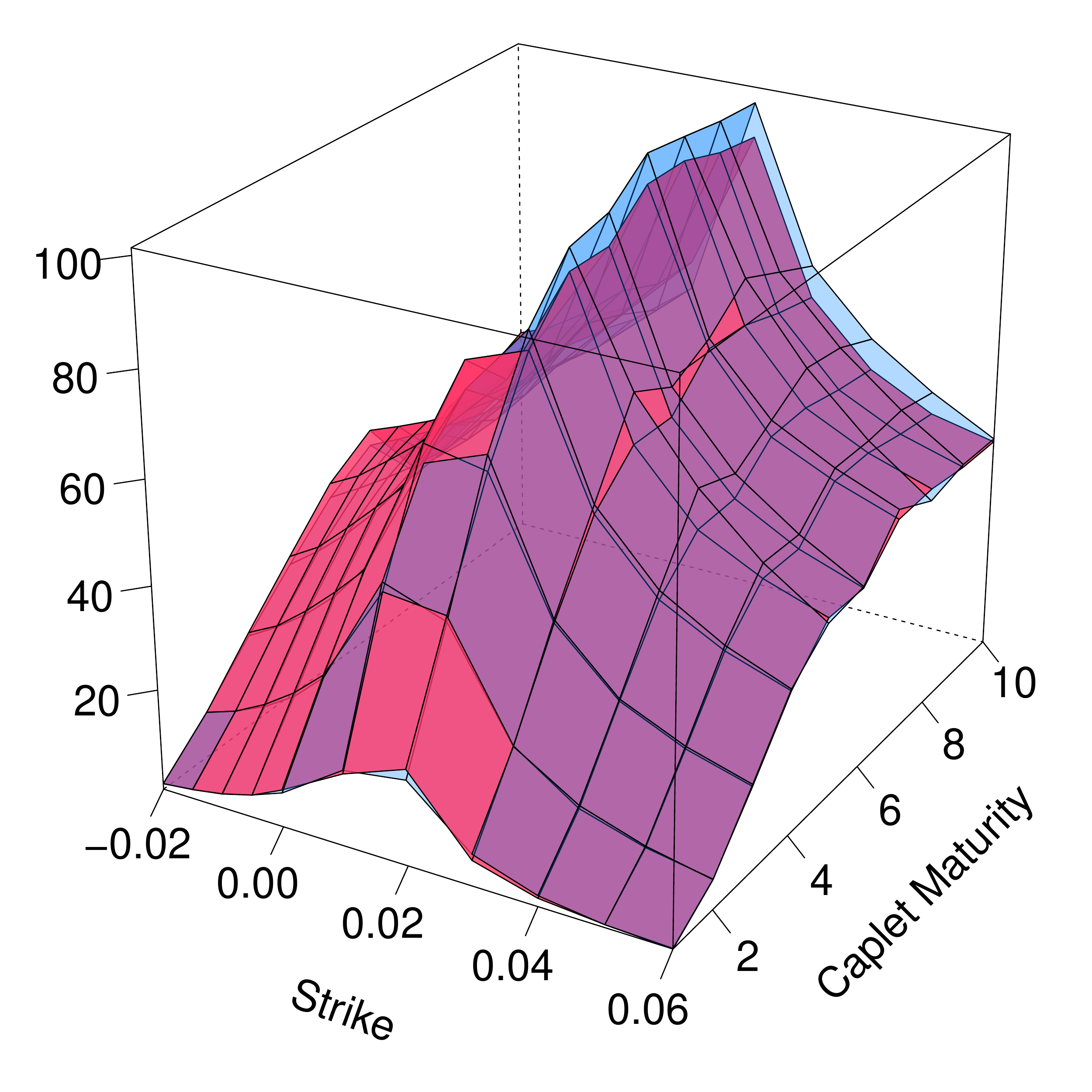}
\captionsetup{singlelinecheck=false, margin = 1.5 cm}
\caption{\small Market and model caplet/floorlet prices in basis points for annual forward inflation with strikes $-2\%$ to $6\%$ and maturities from 1 to 10 years. Market prices are in a transparent blue, while model prices are displayed in red. For strikes between $-2\%$ and $1\%$ prices are quoted for floorlets, for strikes between $2\%$ and $6\%$ prices are quoted for caplets. Market prices are bootstrapped from corresponding cap/floor data.}
\label{fig:inflfit}
\end{figure}
\begin{figure}
\centering
\includegraphics[width=0.55 \textwidth]{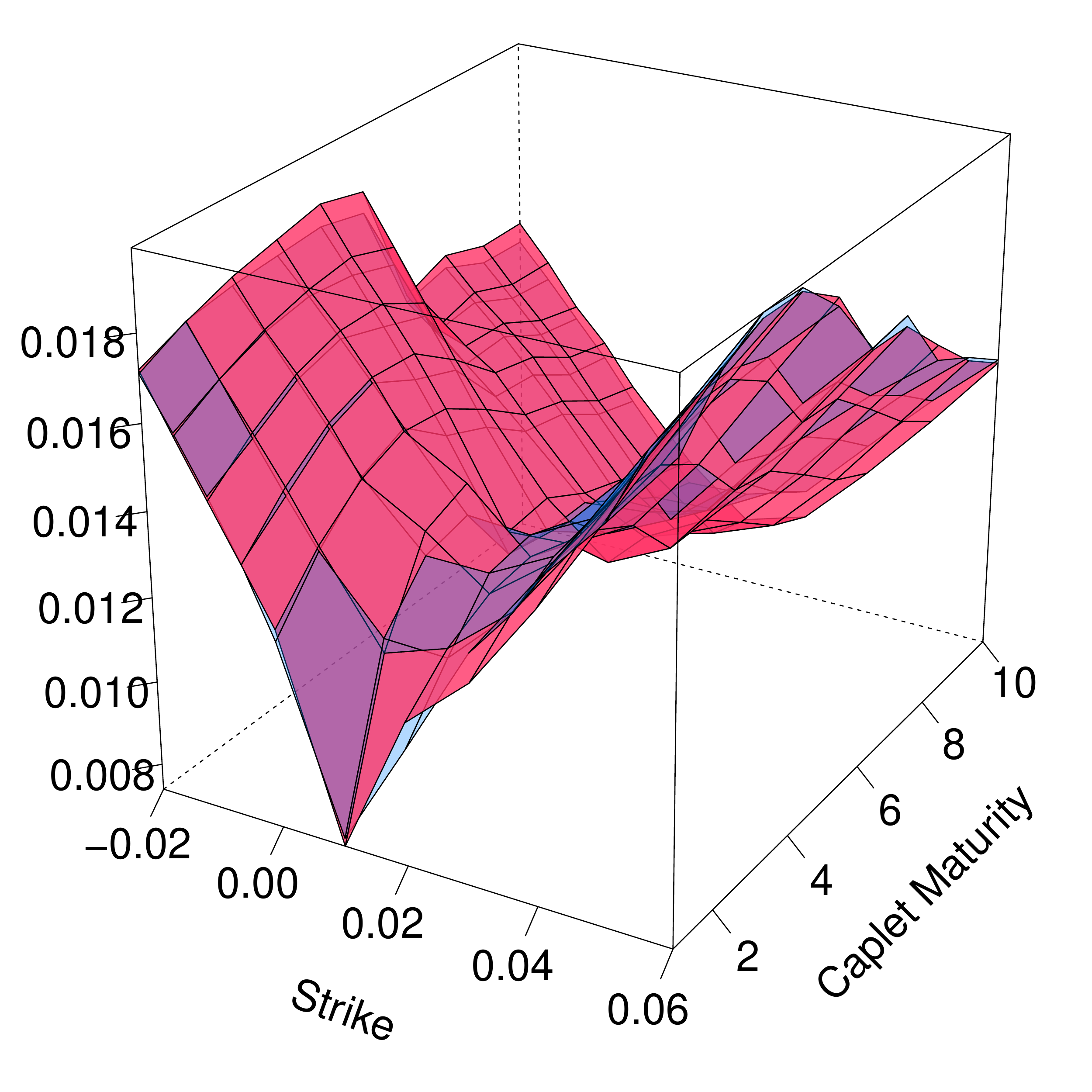}
\captionsetup{singlelinecheck=false, margin = 1.5 cm}
\caption{\small Market and model implied volatilities for annual forward inflation with strikes $-2\%$ to $6\%$ and maturities from 1 to 10 years. Market volatilities are in a transparent blue, while model volatilities are displayed in red.
}
\label{fig:inflfitvol}
\end{figure}

We would also like to display the fit in terms of (shifted lognormal) implied volatilities. Typically one only has quotes on ZCIIS rates,  which do not directly translate to forward inflation rates. However, annual forward inflation rates can be approximated by $F_I(t,T_{k-j},T_k) \approx \I(t,T_k) / \I(t,T_{k-j})-1$ (see figure \ref{fig:forwardinflation}). It is market practice to use this as forward value in the shifted Black formula to calculate approximate market volatilities. The resulting fit in terms of shifted implied volatilities is then displayed in figure \ref{fig:inflfitvol}. This shows that implied volatilities are closely reproduced by the model across all maturities and smiles and that the model is very capable in fitting the observed market data.

\subsection*{Conclusion}
We have introduced a highly tractable inflation market model, where we are able to derive analytical formulas for both types of inflation-indexed swaps. Furthermore inflation caps and floors as well as CPI caps and floors can be calculated with a one-dimensional Fourier inversion formula. Hence prices for liquidly traded inflation derivatives can be calculated quickly and accurately. Additionally the proposed model is able to price classical interest rate derivatives like caps and floors. Using these formulas we are able to calibrate the model to market data. The calibration example shows that the model can be calibrated to inflation market data very accurately.

\begin{appendix} 
\section{Affine processes} \label{sec:affineprocess}
Let  $X = (X_t)_{0 \leq t \leq T}$ be a homogeneous Markov process with values in $D = \R^m_{\geq 0} \times \R^n$ on a measurable space $(\Omega,\mathcal{A})$ with filtration  $(\F_t)_{0 \leq t \leq T}$, with regards to which $X$ is adapted. Denote by $\PM^x$ and $\EV[x]{\cdot}$ the corresponding probability and expectation when $X_0 = x$. 
X is said to be an affine process, if its characteristic function has the form
\begin{equation*} 
\EV[x]{ \e^{u\T X_{t}}}  = \ex{\phi_{t}(u) + \psi_{t}(u)\T x}, \quad u \in \i \R^d, x \in D,
\end{equation*}
where $\phi: [0,T]\times \i \R^d  \rightarrow \C$ and $\psi: [0,T] \times \i \R^d  \rightarrow \C^d$ and $\i \R^d = \{u \in \C^d: \mathrm{Re}(u) = 0 \}.$ 
By homogeneity and the Markov property the conditional characteristic function satisfies
\begin{equation}  \label{eq:affineMomentGen}
 \EV[x]{e^{u\T X_{t}} \vert \F_s} = \ex{\phi_{t-s}(u) + \psi_{t-s}(u)\T X_s}. 
\end{equation}
Accordingly affine processes can also be defined for inhomogeneous Markov processes (see \citet{FI05}). In this case the affine property reads
\begin{equation*} \label{eq:affineMomentGeninhom}
\EV[x]{ \e^{u\T X_t} \vert \F_s} = \ex{\phi_{s,t}(u) + \psi_{s,t}(u)\T X_s}, \quad u \in \i \R^d, x \in D, 
\end{equation*}
with $\phi_{s,t}: \i \R^d  \rightarrow \C$ and $\psi_{s,t}: \i \R^d  \rightarrow \C^d$ for $0 \leq s \leq t$.

$X$ is called an analytic affine process, if $X$ is stochastically continuous and the interior of the set\footnote{$\mathcal{V}$ can be described as the (convex) set, where the extended moment generating function of $X_t$ is defined for all times $t$ and all starting values $x$. By Lemma 4.2 in \citet{KM11} the set $\mathcal{V}$ is in fact equal to the seemingly smaller set 
$\left \{ u \in \C^d:   \exists x \in \mathrm{int}(D): \EV[x]{ \e^{\mathrm{Re}(u)\T X_T}} < \infty \right \}.$}
\begin{align}
\mathcal{V} & := \left \{ u \in \C^d: \sup_{0 \leq s \leq T}  \EV[x]{ \e^{\mathrm{Re}(u)\T X_s}} < \infty \quad  \forall x \in D \right \}, \label{eq:momset}
\end{align}
contains 0
. In this case the functions $\phi$ and $\psi$ have continuous extensions to $\mathcal{V}$, which are analytic in the interior, such that \eqref{eq:affineMomentGen} holds for all $u \in \mathcal{V}$ (see \citet{KR08}).

The class of affine processes includes Brownian motion and more generally all Lévy processes. Since Lévy processes have stationary independent increments, it follows that $\psi_t(u) = u$, while $\phi_t(u) = t \kappa(u)$, where $\kappa$ is the cumulant generating function of the Lévy process. Ornstein-Uhlenbeck processes are further important examples of affine processes. The affine processes used in this work are described at the end of this section. 

The standard reference for affine processes is \citet{DFS03}. There they give a characterization of affine processes, where $\phi$ and $\psi$ are specified as solutions of a system of differential equations
. 
To motivate this consider an affine process $X$. By the tower property for conditional expectations it holds for all $x \in D$
$$ \EV[x]{e^{u X_{t+s}}} = \EV[x]{\EV[x]{e^{u X_{t+s}}\vert \F_s} } =  \EV[x]{e^{\phi_t(u) + \psi_t(u)\T X_s} }. $$
Using equation \eqref{eq:affineMomentGen} it follows that $\phi$ and $\psi$ satisfy the so-called semi-flow equations
\begin{equation}
\begin{aligned}
\phi_{t+s}(u) & = \phi_t(u) + \phi_s(\psi_t(u)),  \qquad & \phi_0(u) & = 0, \\
\psi_{t+s}(u) & = \psi_s(\psi_t(u)),  & \psi_0(u)  & = u.
\end{aligned} \label{eq:semiflow}
\end{equation}
For a stochastically continuous affine process $X$ it was shown in \citet{KST11} that the functions 
\begin{equation*}
F(u) := \left. \frac{\partial}{\partial t} \phi_t(u) \right \vert_{t=0^+}, \qquad R(u) := \left. \frac{\partial}{\partial t} \psi_t(u) \right \vert_{t=0^+}
\end{equation*}
exist\footnote{This was also shown for affine processes with general state spaces in \citet{KST11b} and \citet{CT13}.}.
Rewriting \eqref{eq:semiflow} in terms of difference quotients and letting  $s \rightarrow 0$ we get that $\phi$ and $\psi$ satisfy generalized Ricatti equations
\begin{equation} \label{eq:ricatti}
\begin{aligned}
 \frac{\partial}{\partial t} \phi_t(u) & = F(\psi_t(u)), \qquad \phi_0(u) = 0, \\
 \frac{\partial}{\partial t} \psi_t(u) & = R(\psi_t(u)), \qquad \phi_0(u) = u. 
\end{aligned}
\end{equation}
The functions $F$ and $R$ have a specific form of Levy-Khintchine type as first described in \citet{DFS03}. There it is also shown that for every $F$ and $R$ of this form \eqref{eq:ricatti} has a unique solution. Specifying the functions $F$ and $R$ is an alternative way to specify an affine process. \citet{KM11} give conditions on $F$ and $R$ under which a solution of \eqref{eq:ricatti} defines an analytic affine process. Note that in order to evaluate $\phi$ and $\psi$ one would like to have closed form solutions to the system \eqref{eq:ricatti}, which in general is not the case.

Coupling independent affine processes is very tractable. For two independent affine processes $X$ and $Y$ and all starting values $x, y$ one obtains
\begin{equation} \label{eq:affinecomb}
\EV[(x,y)]{\e^{(u_X, u_Y) \cdot (X_t,Y_t)}} = \EV[(x,y)]{\e^{u_X\T X_t} \e^{u_Y\T Y_t} } = \EV[(x,y)]{\e^{u_X \cdot X_t} } \EV[(x,y)]{\e^{u_Y \cdot Y_t} }.\end{equation}
Hence $(X,Y)$ is an affine process with
\begin{equation} \label{eq:affinecombphipsi}
\begin{aligned}
\phi^{(X,Y)}_t(u_X,u_Y) & = \phi_t^X (u_X)  + \phi_t^Y (u_Y),  \\
\psi^{(X,Y)}_t(u_X,u_Y) & = (\psi_t^X(u_X), \psi_t^Y(u_Y)) .
\end{aligned}
\end{equation}
This fact together with the following Lemma is used in section \ref{sec:numeric}.
\begin{lemma} \label{lem:ufitting}
Let $X$ be an analytic affine processes comprised of $m+n$ independent affine processes, where the first $m$ are nonnegative. For $(u_1, \dots, u_k, v, u_{k+1}, \dots, u_n) \in \mathrm{int}(\mathcal{V}) \cap \R^{m+n}$ define the function
$$f^k(v) := \EV[x]{\e^{(u_1,\dots,u_{k-1},v,u_{k+1},\dots,u_{m+n}) \cdot X_t}}.$$
Then $f^k$ is monotonically increasing if $k \leq m$ and $f^k$ is convex for all $k$. 
\end{lemma}
\begin{proof}
For each $x \in D$ the term inside the expectation is convex in $v$ and monotonically increasing in $v$ if $k \leq m$. This then also holds after taking the expectation. 
\end{proof}

The last part of this section describes the affine processes used in this paper. One classical example is the CIR process, which is the unique solution of
\begin{equation} \label{eq:CIR}
\dd{X_t} = - \lambda (X_t - \theta) \dd{t} + 2\eta \sqrt{X_t} \dd{W_t}, \qquad X_0 = x.
\end{equation}
For this process the functions $\phi$ and $\psi$ are defined for $\mathrm{Re}(u) < \frac{\lambda}{2 \eta^2} (1-\e^{-\lambda t})^{-1}$,
\begin{equation*}
\begin{aligned}
\phi_t(u) & = - \frac{\lambda \theta}{2 \eta^2} \logn{1- \frac{2\eta^2}{\lambda} (1-\e^{-\lambda t})  u}, \\
\psi_t(u) & = \frac{\e^{-\lambda t} u}{1- \frac{2\eta^2}{\lambda} (1-\e^{-\lambda t})  u}. 
\end{aligned}
\end{equation*}
The CIR process almost surely stays nonnegative. It is strictly positive if $\frac{\lambda \theta}{2} > \eta^2$. One can add jumps to this process by adding the differential of a compound Poisson process $L_t$  to the dynamics of $X$. 
\begin{equation} \label{eq:CIRGamOU}
\dd{X_t} = - \lambda (X_t - \theta) \dd{t} + 2\eta \sqrt{X_t} \dd{W_t} + \dd{L_t}, \qquad X_0 = x.
\end{equation}
If $L_t$ has exponentially distributed jumps with expectation values $\frac{1}{\alpha}$ arriving at rate $\lambda \beta$ the functions $\phi$ and $\psi$ are (see \citet{GP13})
\begin{equation*}
\begin{aligned}
\phi_t(u) & = - \frac{\lambda \theta}{2 \eta^2} \logn{1- \frac{2\eta^2}{\lambda} (1-\e^{-\lambda t})  u} - \frac{\lambda \beta}{\lambda - 2 \eta^2 \alpha} \logn{\frac{\alpha-u}{\alpha-u\left(\e^{-\lambda t}+(1-\e^{-\lambda t})\frac{1}{\lambda} 2 \eta^2 \alpha\right)}}, \\
\psi_t(u) & = \frac{\e^{-\lambda t} u}{1- \frac{2\eta^2}{\lambda} (1-\e^{-\lambda t})  u},
\end{aligned}
\end{equation*}
where $$\mathrm{Re}(u) < \min\left\{ \frac{\lambda}{2 \eta^2} (1-\e^{-\lambda t})^{-1}, \alpha \left(\e^{-\lambda t}+(1-\e^{-\lambda t})\frac{1}{\lambda} 2 \eta^2 \alpha\right)^{-1} , \alpha \right\}.$$ Since $L_t$ has only positive jumps this process also stays nonnegative. 
As a third example consider the real-valued affine process defined by
\begin{equation} \label{eq:DGamOUBM}
\dd{X_t} = - \lambda (X_t - \theta) \dd{t} + \sigma \dd{W_t} + \dd{\tilde{L}_t}, \qquad X_0 = x,
\end{equation}
where $\tilde{L}_t$ is a compound Poisson process with positive jumps with mean $\frac{1}{\alpha^+}$ arriving at rate $\lambda \beta^+$ and negative jumps with mean $\frac{1}{\alpha^-}$ arriving at rate $\lambda \beta^-$. The functions $\phi$ and $\psi$ in this case read  (see \citet{MW14})
\begin{align*}
\phi_t(u) & = \frac{\sigma^2 u^2}{4 \lambda} (1-\e^{-2 \lambda t}) + \theta u  (1-\e^{-\lambda t})  + \frac{\beta^+ + \beta^-}{2} \logn{\frac{(\alpha^+-\e^{-\lambda t}u)(\alpha^-+\e^{-\lambda t}u)}{(\alpha^+-u)(\alpha^-+u)}} \\
& + \frac{\beta^+ - \beta^-}{2 } \logn{\frac{(\alpha^+-\e^{-\lambda t}u)(\alpha^-+u)}{(\alpha^+-u)(\alpha^-+\e^{-\lambda t}u)}}, \\
\psi_t(u) & = \e^{-\lambda t} u,
\end{align*}
for $- \alpha^- < \mathrm{Re}(u) < \alpha^+$. 
\end{appendix}

\setlength{\bibsep}{0.0pt}


\begin{thebibliography}{19}
\providecommand{\natexlab}[1]{#1}
\providecommand{\url}[1]{\texttt{#1}}
\expandafter\ifx\csname urlstyle\endcsname\relax
  \providecommand{\doi}[1]{doi: #1}\else
  \providecommand{\doi}{doi: \begingroup \urlstyle{rm}\Url}\fi

\bibitem[Belgrade et~al.(2004)Belgrade, Benhamou, and Koehler]{BB04}
N.~Belgrade, E.~Benhamou, and E.~Koehler.
\newblock {A Market Model for Inflation}.
\newblock \emph{SSRN eLibrary}, 2004.

\bibitem[Brace et~al.()Brace, Gatarek, and Musiela]{BGM97}
A.~Brace, D.~Gatarek, and M.~Musiela.
\newblock The market model of interest rate dynamics.
\newblock \emph{Mathematical Finance}, 7\penalty0 (2).
\newblock ISSN 1467-9965.

\bibitem[Cuchiero and Teichmann(2013)]{CT13}
C.~Cuchiero and J.~Teichmann.
\newblock Path properties and regularity of affine processes on general state
  spaces.
\newblock In \emph{SÃ©minaire de ProbabilitÃ©s XLV}, Lecture Notes in
  Mathematics, pages 201--244. Springer International Publishing, 2013.

\bibitem[Duffie et~al.(2003)Duffie, FilipoviÄ‡, and Schachermayer]{DFS03}
D.~Duffie, D.~FilipoviÄ‡, and W.~Schachermayer.
\newblock Affine processes and applications in finance.
\newblock \emph{Annals of Applied Probability}, 13:\penalty0 984--1053, 2003.

\bibitem[Eberlein et~al.(2010)Eberlein, Glau, and Papapantoleon]{EGP10}
E.~Eberlein, K.~Glau, and A.~Papapantoleon.
\newblock Analysis of fourier transform valuation formulas and applications.
\newblock \emph{Applied Mathematical Finance}, 17\penalty0 (3):\penalty0
  211--240, 2010.

\bibitem[Filipovic(2005)]{FI05}
D.~Filipovic.
\newblock Time-inhomogeneous affine processes.
\newblock \emph{Stochastic Processes and their Applications}, 115\penalty0
  (4):\penalty0 639--659, 2005.

\bibitem[Fleckenstein et~al.(2010)Fleckenstein, Longstaff, and Lustig]{FL10}
M.~Fleckenstein, F.~A. Longstaff, and H.~N. Lustig.
\newblock {Why Does the Treasury Issue TIPS? The TIPS-Treasury Bond Puzzle}.
\newblock \emph{SSRN eLibrary}, 2010.

\bibitem[Grbac and Papapantoleon(2013)]{GP13}
Z.~Grbac and A.~Papapantoleon.
\newblock A tractable libor model with default risk.
\newblock \emph{Mathematics and Financial Economics}, 7\penalty0 (2):\penalty0
  203--227, 2013.
\newblock ISSN 1862-9679.

\bibitem[{Grbac} et~al.(2014){Grbac}, {Papapantoleon}, {Schoenmakers}, and
  {Skovmand}]{GPSS14}
Z.~{Grbac}, A.~{Papapantoleon}, J.~{Schoenmakers}, and D.~{Skovmand}.
\newblock {Affine LIBOR models with multiple curves: theory, examples and
  calibration}.
\newblock \emph{ArXiv e-prints}, May 2014.

\bibitem[Jarrow and Yildirim(2003)]{JY03}
R.~Jarrow and Y.~Yildirim.
\newblock Pricing treasury inflation protected securities and related
  derivatives using an {HJM} model.
\newblock \emph{Journal of Financial and Quantitative Analysis}, 38\penalty0
  (2):\penalty0 337--359, 2003.

\bibitem[Keller-Ressel(2008)]{KR08}
M.~Keller-Ressel.
\newblock \emph{Affine processes - Theory and application in finance}.
\newblock {PhD}-thesis, Wien University of Technology, 2008.

\bibitem[{Keller-Ressel} and {Mayerhofer}(2011)]{KM11}
M.~{Keller-Ressel} and E.~{Mayerhofer}.
\newblock Exponential moments of affine processes.
\newblock \emph{ArXiv e-prints}, Nov. 2011.

\bibitem[Keller-Ressel et~al.(2001)Keller-Ressel, Teichmann, and
  Schachermayr]{KST11}
M.~Keller-Ressel, J.~Teichmann, and W.~Schachermayr.
\newblock Affine processes are regular.
\newblock \emph{Journal of Probability Theory and Related Fields}, 151\penalty0
  (3-4):\penalty0 591--611, 2001.

\bibitem[Keller-Ressel et~al.(2013{\natexlab{a}})Keller-Ressel, Papapantoleon,
  and Teichmann]{KPT11}
M.~Keller-Ressel, A.~Papapantoleon, and J.~Teichmann.
\newblock The affine {LIBOR} models.
\newblock \emph{Mathematical Finance}, 23\penalty0 (4):\penalty0 627--658,
  2013{\natexlab{a}}.

\bibitem[Keller-Ressel et~al.(2013{\natexlab{b}})Keller-Ressel, Schachermayer,
  and Teichmann]{KST11b}
M.~Keller-Ressel, W.~Schachermayer, and J.~Teichmann.
\newblock Regularity of affine processes on general state spaces.
\newblock \emph{Electron. J. Probab.}, 18:\penalty0 no. 43, 1--17,
  2013{\natexlab{b}}.

\bibitem[Mercurio(2005)]{ME05}
F.~Mercurio.
\newblock Pricing inflation-indexed derivatives.
\newblock \emph{Journal of Quantitative Finance}, 5\penalty0 (3):\penalty0
  289--302, 2005.

\bibitem[Mercurio and Moreni(2006)]{MM06}
F.~Mercurio and N.~Moreni.
\newblock Inflation with a smile.
\newblock \emph{Risk magazine}, 2006.

\bibitem[Mercurio and Moreni(2009)]{MM09}
F.~Mercurio and N.~Moreni.
\newblock A multi-factor {SABR} model for forward inflation rates.
\newblock \emph{SSRN eLibrary}, 2009.

\bibitem[M\"uller and Waldenberger(2015)]{MW14}
W.~M\"uller and S.~Waldenberger.
\newblock Affine {LIBOR} models driven by real-valued affine processes.
\newblock arXiv:1503.00864, 2015.

\end{thebibliography}

\end{document}